\documentclass[conference, 10pt, twocolumn]{IEEE_Transaction}

\usepackage{tikz}
\usetikzlibrary{automata,positioning}
\usepackage[T1]{fontenc}
\usepackage{amsmath,amsfonts,amssymb,amsthm,amsbsy}
\usepackage{float}
\usepackage{textcomp}
\usepackage{cite}
\usepackage{bm}
\usepackage{algorithm} 
\usepackage{xcolor}
\usepackage{arydshln}
\usepackage{multicol}
\makeatletter
\newsavebox\myboxA
\newsavebox\myboxB
\newlength\mylenA
\usepackage{balance}
\newcommand*\xoverline[2][0.70]{%
    \sbox{\myboxA}{$\m@th#2$}%
    \setbox\myboxB\null
    \ht\myboxB=\ht\myboxA%
    \dp\myboxB=\dp\myboxA%
    \wd\myboxB=#1\wd\myboxA
    \sbox\myboxB{$\m@th\overline{\copy\myboxB}$}
    \setlength\mylenA{\the\wd\myboxA}
    \addtolength\mylenA{-\the\wd\myboxB}%
    \ifdim\wd\myboxB<\wd\myboxA%
       \rlap{\hskip 0.5\mylenA\usebox\myboxB}{\usebox\myboxA}%
    \else
        \hskip -0.5\mylenA\rlap{\usebox\myboxA}{\hskip 0.5\mylenA\usebox\myboxB}%
    \fi}
\makeatother
\usepackage[algo2e]{algorithm2e} 
\usepackage[noend]{algpseudocode}
\newcommand\blfootnote[1]{%
	\begingroup
	\renewcommand\thefootnote{}\footnote{#1}%
	\addtocounter{footnote}{-1}%
	\endgroup}
\usepackage{lipsum}
\usepackage{authblk}
\makeatletter
\renewcommand\AB@affilsepx{ \protect\Affilfont} 
\makeatother
\makeatletter
\def\BState{\State\hskip-\ALG@thistlm}
\makeatother
\usepackage{mathtools}

\theoremstyle{remark}

\usepackage[section]{placeins}
\usepackage{amsmath}
\newcommand{\argminD}{\arg\!\min}
\usepackage[cmintegrals]{newtxmath}

\newtheorem{theorem}{\textbf{Theorem}}
\renewenvironment{proof}{{\textbf{Proof.}}}{}
\begin{document}

	\title{Bit-Interleaved Coded Multiple Beamforming with Perfect Coding in  Massive MIMO Systems }

	\small
	\author{Sadjad Sedighi, \textit{Student Member, IEEE,}}
	\author{Ender Ayanoglu, \textit{Fellow, IEEE}}
	\affil{CPCC, Dept of EECS, UC Irvine, Irvine, CA, USA,}
	\normalsize

	\maketitle
	\begin{abstract}
	This letter investigates bit-interleaved coded multiple beamforming (BICMB) with perfect coding in millimeter-wave (mm-wave) massive multiple-input multiple-output (MIMO) systems to achieve both maximum diversity gain and multiplexing gain. Using perfect coding with BICMB enables us to do this. We show that by using BICMB and perfect coding, the diversity gain becomes independent from the number of transmitted data streams and the number of antennas in each remote antenna unit (RAU) at the transmitter and the receiver. The assumption is that the perfect channel state information (CSI) is known at both the transmitter and the receiver and the number of antennas goes to infinity. This latter assumption can be relaxed by a large number of antennas in each RAU, similar to the case for all massive MIMO research. Simulation results show that when the perfect channel state information assumption is satisfied, the use of BICMB with perfect coding results in the diversity gain values predicted by the analysis.
	\end{abstract}
	\blfootnote{This work was partially supported by NSF under Grant No. 1547155.}
	\section{Introduction}
	Diversity gain analysis of a millimeter-wave (mm-wave) massive multiple-input multiple-output (MIMO) system  with distributed antenna-subarray architecture was first studied in \cite{Dian2018J}. The diversity gain calculated in \cite{Dian2018J} depends on the number of transmitted data streams in the system. This means by increasing the number of transmitted data streams, the diversity gain decreases. Furthermore, the diversity gain in \cite{Dian2018J} can be increased simply by increasing the number of antenna subarrays. The system model used in \cite{Xiao2015} is co-located and the maximum achievable diversity gain is the number of multipath components (MPC), i.e., $L$. This maximum value for the diversity gain can be achieved by only transmitting one data stream through the system.
	Diversity gain analysis for the mm-wave MIMO systems is studied in \cite{Ruan2019,Xiao2013}. In \cite{Elganimi2018}, a combination of Space-Time Block Coded Spatial Modulation with Hybrid Analog-Digital Beamforming is used to achieve the full diversity gain for mm-wave MIMO systems. 
	
	Bit-interleaved coded modulation (BICM) was first introduced to increase the code diversity \cite{Caire1998,Zehavi1992}. Later on, bit-interleaved coded multiple beamforming (BICMB) was used to achieve full diversity gain and full multiplexing gain in MIMO systems \cite{Akay2007,Park2009,Sengul2009}. In this method, different codewords are interleaved among different subchannels with different diversity orders. In \cite{Xiao2015}, iterative eigenvalue decomposition (EVD) is used to find the antenna weight vectors such that the full diversity gain is achieved for mm-wave massive MIMO systems. To overcome this diversity degradation, in \cite{Sedighi2019}, we proved that by using BICMB in a single-user mm-wave massive MIMO system with distributed antenna-subarray architecture both full diversity gain and full multiplexing gain can be achieved.  Furthermore, unlike \cite{Xiao2015}, the diversity gain can be increased by increasing the number of remote antenna units (RAUs) at both the transmitter and the receiver side.
	 
	 Perfect space-time block codes (PSTBC) was studied in \cite{Oggier2006,Elia2005,Berhuy2009} to achieve full rate and full diversity in any dimension. However, dimensions 2, 3, 4 and 6 are the only dimensions yield to increasing the coding gain. In \cite{Li2010}, perfect coding with multiple beamforming is used to achieve full diversity and full multiplexing in a MIMO system with less decoding complexity than a system employing the PSTBC and full precoded multiple beamforming (FPMB). In \cite{Li2012}, channel coding is added to the perfect coding and diversity gain analysis is carried out to prove that BICMB with perfect coding achieves the full diversity order.

In this work, we use BICMB with perfect coding to achieve full diversity gain and full multiplexing gain in mm-wave massive MIMO systems. The diversity analysis is carried out. We show that by using perfect coding in addition to convolutional coding, the diversity gain becomes independent from the number of transmitted data streams.

	The remainder of this paper is organized as follows: In Section II, the system model and the description of BICMB with perfect coding is given. In Section III, the diversity gain analysis is provided. In Section IV, the decoding technique is studied. Section V provides the simulation results. Finally, conclusions are presented in Section VI.
	
	\textit{Notation:}
Boldface upper and lower case letters denote matrices and column vectors, respectively.	The  Hamming distance
	between any two codewords in a convolutional code is defined as the free distance $d_{H}$. The symbols $(.)^H, (.)^*,(\bar{.})$ and $\forall$ denote the Hermitian, complex conjugate, binary complement, and for all, respectively. $\mathcal{CN} (0, 1)$ denotes a circularly symmetric complex Gaussian random distribution with zero mean and unit variance. The expectation operator is denoted by $\mathbb{E}\left[.\right]$. Finally, diag$\{a_1 , a_2 , \dots , a_N \}$ stands for a diagonal matrix with diagonal elements $\{a_1 , a_2 ,\dots , a_N \}$.

	\section{System Model}
	
	One can approximate the average probability of bit error (BER) $P_E$ at high SNR for both coded and uncoded system as \cite{Proakis2007,Wang2003}
    \begin{align}
        P_E \approx (G_c \bar{\gamma})^{-G_d},
    \end{align}
    where $G_c$ and $G_d$ are defined as coding gain and diversity gain, respectively. Average SNR is shown by $\bar{\gamma}$. In a log-log scale, diversity gain $G_d$ determines the slope of the BER versus the average SNR curve in high SNR regime. Furthermore, changing $G_c$ leads to shift of the curve in SNR relative to a benchmark BER curve of $(\bar{\gamma}^{-G_d})$. 
    
	We consider a single-user mm-wave massive MIMO scenario shown in Fig. \ref{su_fc}, where the transmitter is equipped with $M_t$ RAUs, $N_t$ antennas at each RAU, and $N_t^{\text{RF}}$ RF chains. The receiver has $M_r$ RAUs, $N_r$ antennas at each RAU, and $N_r^{\text{RF}}$ RF chains. The transmitter sends $N_s=D$ data streams to the receiver. These data streams are generated as follow. First the bit codeword $\mathbf{c}$ is generated through a convolutional encoder with code rate $R_c$. Then a random bit-interleaver is used to generate an interleaved sequence. The output of the interleaver is modulated by M-quadrature amplitude modulation (QAM).	We define a one-to-one mapping from $\mathbf{X}_k = \left[\mathbf{x}_{1,k},\dots,\mathbf{x}_{D,k}\right]$ to $\mathbf{Z}_k$ as $\mathbf{Z}_k=\mathbb{M}\left\{\mathbf{X}_k\right\}$ where $\mathbb{M}$ denotes the PSTBC codeword generating function. A PSTBC codeword, i.e., $\mathbf{Z}_k$ is generated by using $D^2$ consecutive complex-valued scalar symbols \cite{Li2012}
	\begin{align}\label{stream}
 \mathbf{Z}_k=\mathbb{M}\left\{\mathbf{X}_k\right\}= \sum_{v=1}^{D} \text{diag}(\mathbf{G}\mathbf{x}_{v,k})\mathbf{E}^{v-1},
	\end{align}
	where $\mathbf{G}$ is an $D \times D$ unitary matrix \cite{Oggier2006}, $\mathbf{x}_{v,k}$ is an $D\times 1$ vector whose elements are the $v$th $D$ input modulated scalar symbols and $D \in \{2,3,4,6 \}$. Matrix $\mathbf{E}$ is defined as
	\begin{align}
	    \mathbf{E} = \begin{bmatrix}
    0 & 1 & 0 & \cdots & 0 & 0\\
    0 & 0 & 1 & \cdots & 0 & 0 \\
    \vdots & \vdots & \vdots &\ddots & \ddots & \vdots \\
    0 & 0 & 0 &\cdots & 0 & 1 \\
    g & 0 & 0 &\cdots & 0 & 0
\end{bmatrix},
	\end{align}
	where $g=\left\{
    \begin{array}{ll}
        i, & D=2,4, \\
        e^{\frac{2\pi}{3}}, & D=3,\\
        -e^{\frac{2\pi}{3}}, & D=6.
    \end{array}
\right.$ 

As it can be seen from Fig. \ref{su_fc}, the complex-valued matrix $\mathbf{F}_{\text{BB}} \in \mathbb{C}^{N_t^{RF} \times N_s}$ is used for preprocessing at the baseband. A set of $M_tN_t$ phase shifters is applied to the output of each RF chain. As a result of this process, different beams are formed in order to transmit the RF signals. We can model this process with  complex-valued matrix $\mathbf{F}_{\text{RF}} \in \mathbb{C}^{M_t N_t \times N_t^{RF} }$. Note that in this work $M_t=M_r=N_s=D$.
		\begin{figure}[!t]
		\centering
        \includegraphics[width=.5\textwidth]{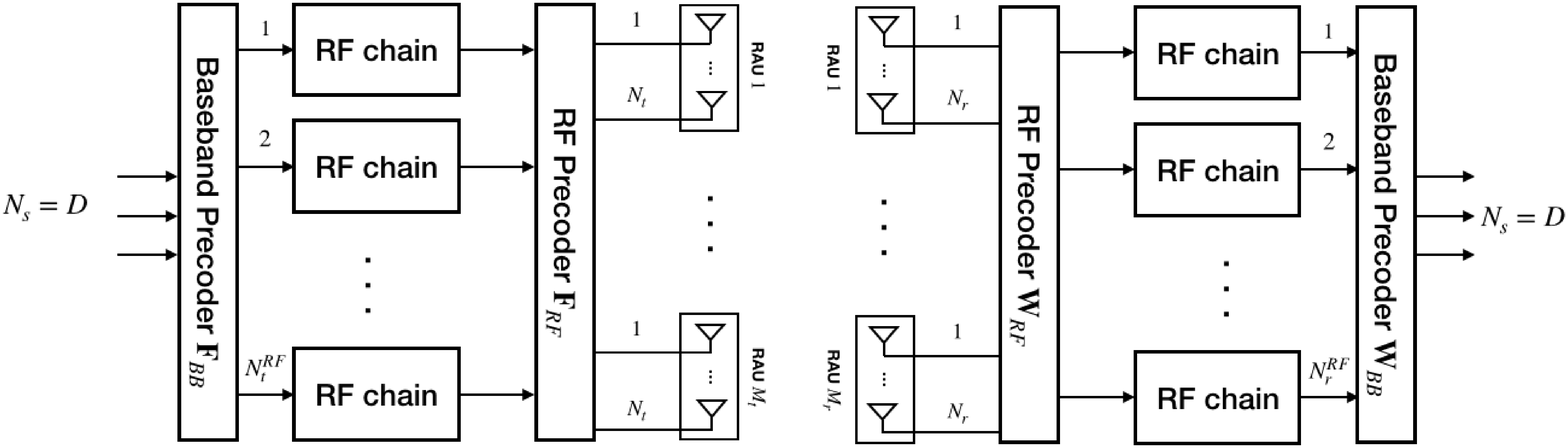}
		\caption{Block diagram of a mm-wave massive MIMO system with distributed antenna sub-arrays.}
		\label{su_fc}
	\end{figure}
	
	By assuming a narrowband flat fading channel model, we write the $M_rN_r\times M_tN_t$ channel matrix $\mathbf{H}$ as  
	\begin{align} \label{channel_su}
	\mathbf{H}=
	\begin{bmatrix}
	\sqrt{\beta_{11}}\mathbf{H}_{11}&\sqrt{\beta_{12}}\mathbf{H}_{12} & \dots& \sqrt{\beta_{1M_t}}\mathbf{H}_{1M_t} \\
	\sqrt{\beta_{21}}\mathbf{H}_{21}&\sqrt{\beta_{22}}\mathbf{H}_{22} & \dots& \sqrt{\beta_{2M_t}}\mathbf{H}_{2M_t} \\
	\vdots & \vdots & \ddots & \vdots  \\
	\sqrt{\beta_{M_r1}}\mathbf{H}_{M_r1} & \sqrt{\beta_{M_r2}}\mathbf{H}_{M_r2} & \dots & \sqrt{\beta_{M_rM_t}}\mathbf{H}_{M_rM_t}
	\end{bmatrix},
	\end{align}
	where each $\mathbf{H}_{ij}$ is the MIMO channel between the $i$th RAU at the receiver and the $j$th RAU at the transmitter. Also, $\beta_{ij}$ is a real-valued nonnegative number and represents the large-scale fading effect between the $i$th RAU at the receiver and $j$th RAU at the transmitter.
	Note that in this work, we use Saleh-Valenzuela model for each subchannel $\mathbf{H}_{ij}$ \cite{HXu2002,Standard,Sedighi2019}. For the sake of simplicity, each scattering cluster is assumed to contribute a single propagation path. The subchannel matrix $\mathbf{H}_{ij}$ is given by
	\begin{align}\label{Hij}
	\mathbf{H}_{ij}=\sqrt{\frac{N_tN_r}{L_{ij}}} \sum_{l=1}^{L_{ij}}\alpha_{ij}^l\mathbf{a}_r(\theta_{ij}^{l}) \mathbf{a}_t^H(\phi_{ij}^{l}),
	\end{align}
	where $L_{ij}$ is the number of propagation paths and $\alpha_{ij}^{l}$ is the complex-gain of the $l$th ray which follows $\mathcal{CN}(0,1)$, $\theta_{ij}^l \in [0,2\pi]$,  $\phi_{ij}^l \in [0,2\pi]$, $\forall i,j,l$, and the vectors $\mathbf{a}_r(\theta_{ij}^{l})$ and $ \mathbf{a}_t(\phi_{ij}^{l})$ are the normalized array response at the receiver and transmitter, respectively. In particular, this paper adopts a  uniform linear array (ULA) where $\mathbf{a}_r(\theta_{ij}^{l})$ and $ \mathbf{a}_t(\phi_{ij}^{l})$ are modeled as
	
	\begin{align}
	\mathbf{a}_r(\theta_{ij}^{l})=\frac{1}{\sqrt{N_r}}\left[1, e^{j\frac{2\pi}{\lambda}d\sin(\theta_{ij}^{l})},\dots, e^{j(N_r-1)\frac{2\pi}{\lambda}d\sin(\theta_{ij}^{l})}\right]^T,
	\end{align}
	\begin{align}
	\mathbf{a}_t(\phi_{ij}^{l})=\frac{1}{\sqrt{N_t}}\left[1, e^{j\frac{2\pi}{\lambda}d\sin(\phi_{ij}^{l})},\dots, e^{j(N_t-1)\frac{2\pi}{\lambda}d\sin(\phi_{ij}^{l})}\right]^T,
	\end{align}
	where $\lambda$ is the transmission wavelength, and $d$ is the antenna spacing.
	
    The processed signal at the $k$th PSTBC codeword is
	\begin{align}\label{rec_sig}
	\mathbf{Y}_k=\mathbf{W}_{\text{BB}}^H\mathbf{W}_{\text{RF}}^H\mathbf{H}\mathbf{F}_{\text{RF}}\mathbf{F}_{\text{BB}}\mathbf{Z}_k+\mathbf{W}_{\text{BB}}^H\mathbf{W}_{\text{RF}}^H\mathbf{n}_k,
	\end{align}
	where $\mathbf{Y}_k$ is an $D\times D$ complex-valued matrix, $\mathbf{n}_k$ is an $M_rN_r\times1$ vector consisting of i.i.d. $\mathcal{CN}(0,N_0)$ noise samples, where $N_0=\frac{N_t}{SNR}$ and $SNR$ is defined as the signal-to-noise ratio (SNR), $\mathbf{W}_{\text{RF}}$ is the $M_rN_r\times N_r^{RF}$ RF combining matrix, and $\mathbf{W}_{\text{BB}}$ is the $N_r^{(\text{RF})}\times N_s$ baseband combining matrix.
	
	A solution based on singular value decomposition (SVD) of the channel matrix $\mathbf{H}=\mathbf{U}\mathbf{\Lambda}\mathbf{V}^H$ can be derived for the beamforming matrices \cite{Ayach2012,Sedighi2019}. By utilizing the optimal precoder and combiner, one can write (\ref{rec_sig}) as
	\begin{align}\label{proc_sig}
	\mathbf{Y}_k=\mathbf{\Lambda}\mathbf{Z}_k+\tilde{\mathbf{n}}_k,
	\end{align}
	where $\tilde{\mathbf{n}}_k=\mathbf{U}_{(1:D)}^H\mathbf{n}_k,$ and $\mathbf{U}_{(1:D)}$ is the first $D$ columns of the unitary matrix $\mathbf{U}$.

 We model PSTBC codeword sequence as $k' \rightarrow (k,(m,n),j)$, where $k'$ represents the original ordering of the coded bits $c_{k'}$, $(k,(m,n),j)$ are the index of the PSTBC codewords, the symbol position in $\mathbf{X}_k$, and the bit position on the label of the scalar symbol $x_{(m,n),k}$, respectively. 
		We define $\chi_b^j$ as the subset of all signals $x\in\chi$. Note that the label has the value $b\in\{0,1\}$ in position $j$. 
		
		The maximum likelihood (ML) bit metrics for (\ref{rec_sig}) can be written as
		\begin{align}\label{ML_metric}
		    \gamma^{(m,n),j}(\mathbf{Y}_k,c_{k'}) = \min_{\mathbf{X} \in \eta_{c_{k'}}^{(m,n),j} }|| \mathbf{Y}_k-\Lambda\mathbb{M}\{\mathbf{X}\}||^2,
		\end{align}
where $\eta_{c_{k'}}^{(m,n),j} $ is defined as
\begin{align}
\eta_{c_{k'}}^{(m,n),j} = \{ \mathbf{X}: x_{(u,v)=(m,n)}\in \chi_b^j, \text{ and } x_{(u,v)\neq(m,n)} \in \chi \}.
\end{align}

The ML decoder at the receiver makes the decisions according to the rule
\begin{align}
    \hat{\mathbf{c}}=\argminD_{\mathbf{c}}\sum_{k'}\gamma^{(m,n),j}(\mathbf{Y}_k,c_{k'}).
\end{align}
	\section{Diversity Gain Analysis}
    
 In this section, diversity gain is examined for mm-wave massive MIMO systems with distributed antenna-subarray architecture employing BICMB with perfect coding. We show that the diversity gain becomes independent from the number of transmitted streams, comparing to \cite{Dian2018J}.
 \begin{theorem}
			Suppose that $N_r\rightarrow \infty$ and $N_t \rightarrow \infty$. Then by utilizing BICMB with perfect coding, mm-wave massive MIMO systems with distributed antenna-subarray architecture can achieve a diversity gain of
		\begin{align}\label{su_dg}
		G_d=\frac{\left(\sum_{i,j}\beta_{ij}\right)^2}{\sum_{i,j}\beta_{ij}^2L_{ij}^{-1}}
		\end{align}
			$\text{ for } i=1,\dots,M_r \text{ and } j=1,\dots,M_t$.
	\end{theorem}
		\begin{proof}
		\end{proof}
 Assume that codeword $\mathbf{c}$ is transmitted and codeword $\hat{\mathbf{c}}$ is detected. Then one can write the pairwise error probability (PEP) of $\mathbf{c}$ and $\hat{\mathbf{c}}$ as
	\begin{align}{\label{su_PEP1}}
	    P(\mathbf{c}\rightarrow\hat{\mathbf{c}}|\mathbf{H}) = P\left( \sum_{k'}  ||\mathbf{Y}_k-\Lambda \tilde{\mathbf{Z}}\}||^2 
	     \geq \sum_{k'} ||\mathbf{Y}_k-\Lambda \hat{\mathbf{Z}}\} ||^2 \; | \;  \mathbf{H}\right),
	\end{align}
	where $\tilde{\mathbf{Z}} = \mathbb{M}\{\tilde{\mathbf{X}}\}$, $\hat{\mathbf{Z}} = \mathbb{M}\{\hat{\mathbf{X}}\}$, $\tilde{\mathbf{X}} = \argminD_{\mathbf{X} \in \eta_{c_{k'}}^{(m,n),j}} ||\mathbf{Y}_k-\Lambda \mathbb{M}\{\mathbf{X}\}||^2 $, and $\hat{\mathbf{X}} = \argminD_{\mathbf{X} \in \eta_{\bar{c}_{k'}}^{(m,n),j}} ||\mathbf{Y}_k-\Lambda \mathbb{M}\{\mathbf{X}\}||^2 $.
	Since the bit metrics corresponding to the same coded bits between the pairwise errors are the same and $||\mathbf{Y}_k-\mathbf{\Lambda}\mathbf{Z}_k||^2 \geq ||\mathbf{Y}_k-\mathbf{\Lambda}\hat{\mathbf{Z}}_k||$, (\ref{su_PEP1}) is  upper-bounded by
		\begin{align}{\label{su_PEP2}}
	    P(\mathbf{c}\rightarrow\hat{\mathbf{c}}|\mathbf{H})  \leq P\left(\xi \geq \sum_{k',d_H}||\mathbf{\Upsilon}||^2 \right)
	\end{align}
	where ${\sum_{k',d_H}}$ is the summation of the $d_H$ values corresponding to the different coded bits between the bit codewords, $\mathbf{\Upsilon} =  \mathbf{\Lambda}(\mathbf{Z}-\hat{\mathbf{Z}}_k)$, and $\xi =- \sum_{k',d_H}\text{tr}\left(\mathbf{\Upsilon}^H\mathbf{n}_k+\mathbf{n}_k^H\mathbf{\Upsilon}\right)$. Since $\xi\sim \mathcal{CN}\left(0,2N_0\sum_{k',d_H}||\mathbf{\Upsilon}||^2\right)$, (\ref{su_PEP2}) is replaced by the $Q$ function as 
			\begin{align}{\label{su_PEP3}}
	    P(\mathbf{c}\rightarrow\hat{\mathbf{c}}|\mathbf{H})  \leq Q\left(\sqrt{\frac{\sum_{k',d_H}||\mathbf{\Upsilon}||^2}{2N_0} } \right)
	\end{align}
	By using the inequality of the $Q$ function where it is upper-bounded as $Q(x) \leq \frac{1}{2}e^{-\frac{x^2}{2}}$, the average PEP in is upper-bounded as
	\begin{align}\label{su_PEP4}
	    P(\mathbf{c}\rightarrow\hat{\mathbf{c}}) =\mathbb{E}[ P(\mathbf{c}\rightarrow\hat{\mathbf{c}}|\mathbf{H})]\leq \frac{1}{2}\mathit{E}\left[\exp\left(-\frac{\sum_{k',d_H} ||\mathbf{\Upsilon} ||^2}{4N_0}\right)\right].
	\end{align}
	
	By using \cite{Li2010}, we can rewrite (\ref{su_PEP4}) as
	\begin{align}
	  P(\mathbf{c}\rightarrow\hat{\mathbf{c}}) =\frac{1}{2}\mathit{E}\left[\exp\left(-\frac{\sum_{u=1}^D \lambda^2_u \zeta_u }{4N_0}\right)\right],
	\end{align}
	where $\zeta_u = \sum_{k',d_H}\rho_{u,k}$ and $\rho_{u,k}=\sum_{v=1}^D|\mathbf{g}_u^T(\mathbf{x}_{v,k}-\hat{\mathbf{x}}_{v,k})|^2$.
	
	By defining $L_t =\sum_{i,j}L_{ij}$ as the rank of the channel matrix $\mathbf{H}$, i.e., the number of singular values of the channel matrix $\mathbf{H}$, we can write
			\begin{align}\label{su_ineq}
		\frac{\left(\zeta_{\text{min}}\sum_{u=1}^{L_t}\lambda_u^2\right)}{L_t} \leq 
		\frac{\left(\zeta_{\text{min}}\sum_{u=1}^{D}\lambda_u^2\right)}{D} \leq
		\frac{\left(\sum_{u=1}^{D}\lambda_u^2\zeta_u\right)}{D}.
		\end{align}
	
		One can define
		\begin{align} \label{theta}
		\Theta \triangleq \sum_{u=1}^{L_t}\lambda_u^2=||\mathbf{H}||_F^2=\sum_{i=1}^{M_r}\sum_{j=1}^{M_t}\beta_{ij}||\mathbf{H}_{ij}||_F^2.
		\end{align}
		
		Theorem 3 in \cite{Ayach2012} implies that the singular values of $\mathbf{H}_{ij}$ converge to $\sqrt{\frac{N_r N_t} {L_{ij}}}\left|\alpha_{l}^{ij}\right|$ in descending order. By using the singular values of $\mathbf{H}_{ij}$, (\ref{theta}) can be rewritten as
		
		\begin{align}\label{Psi_def}
		\Theta=\sum_{i=1}^{M_r}\sum_{j=1}^{M_t}\beta_{ij}||\mathbf{H}_{ij}||_F^2= N_r N_t\sum_{i=1}^{M_r}\sum_{j=1}^{M_t}\underbrace{\frac{\beta_{ij}}{L_{ij}}\sum_{l=1}^{L_{ij}}\left|\alpha_{ij}^l\right|^2}_{\Psi_{ij}}.
		\end{align}
		
		Note that the random variable $\sum_{l=1}^{L_{ij}}\left|\alpha_{ij}^l \right|$ has a $\chi$-squared distribution with $2L_{ij}$ degrees of freedom, or equivalently a Gamma distribution with shape $L_{ij}$ and scale 2, i.e., $\mathcal{G}(L_{ij},2)$. Then, since $\beta_{ij} L_{ij}^{-1}>0$, $\Psi_{ij} \sim \mathcal{G}(L_{ij},2\beta_{ij}L_{ij}^{-1})$ \cite{Hogg1978}. We use the Welch-Satterthwaite equation to calculate an approximation to the degrees of freedom of $\Theta$ (i.e., shape and scale of the Gamma distribution) which is a linear combination of the independent random variables $\Psi_{ij}$ \cite[p.4.1-1]{Satterth1946},\cite{Massey}
		\begin{align}\label{shape}
		\kappa&=\frac{\left(\sum_{i,j}\theta_{ij}k_{ij}\right)^2}{\sum_{i,j}\theta_{ij}^2k_{ij}}=\frac{\left(\sum_{i,j}\beta_{ij}\right)^2}{\sum_{i,j}\beta_{ij}^2L_{ij}^{-1}},\\\label{scale}
		\theta&=\frac{\sum_{i,j}\theta_{ij}^2k_{ij}}{\sum_{i,j}\theta_{ij}k_{ij}}=\frac{\sum_{i,j}\beta_{ij}^2L_{ij}^{-1}}{\sum_{i,j}\beta_{ij}}.
		\end{align}
		
	By using (\ref{su_ineq}), we can upper-bound the PEP in (\ref{su_PEP4}) by
		\begin{align}\label{PEP_exp2_single_user}
		P(\mathbf{c}\rightarrow\hat{\mathbf{c}})\leq\frac{1}{2} \mathit{E}\left[\text{exp}\left(\frac{-\zeta_{\text{min}}D}{4N_0L_t}\Theta\right)\right],
		\end{align}
		which is the definition of the moment generating function (MGF)\cite{Bulmer1965} for the random variable $\Theta$. By using the definition, (\ref{PEP_exp2_single_user}) can be written as
		\begin{align}
		P(\mathbf{c}\rightarrow\hat{\mathbf{c}})\leq&\frac{1}{2} \left(1+ \theta \frac{\zeta_{\text{min}}DN_t}{4L_t}SNR\right)^{-\kappa}\\
		\approx&\frac{1}{2}\left(\theta\frac{\zeta_{\text{min}}DN_t}{4L_t}SNR\right)^{-\kappa} 
		\end{align}
		for high SNR.  
		
		Hence, BICMB with perfect coding achieves full diversity order of
		\begin{align}\label{su_Gd}
		G_d=\kappa=\frac{\left(\sum_{i,j}\beta_{ij}\right)^2}{\sum_{i,j}\beta_{ij}^2L_{ij}^{-1}}
		\end{align}
	which is independent of the number of spatial streams transmitted.
	
		\newtheorem{remarks}{\textbf{Remark}}
			\begin{remarks}
By assuming that $L_{ij}=L$ and $\beta_{ij}=\beta$ for any $i \in \{1,\dots,M_r\}$ and $j \in \{1,\dots,M_t\}$, it can be seen easily that the mm-wave massive MIMO system with distributed antenna-subarray architecture can achieve a diversity gain
		\begin{align} \label{su_Gd_remark1}
	    G_d=L_t=M_rM_tL.
	    \end{align}
	  One can compare this result with the diversity gain calculated for the single-user scenario in \cite{Dian2018J}. As it can be seen, similar to \cite{Sedighi2019}, the full diversity gain and full multiplexing gain is achieved in this paper.
	\end{remarks}
	\section{Decoding}
By replacing (\ref{stream}) in (\ref{proc_sig}), we can rewrite (\ref{proc_sig}) and show that each element of $\mathbf{\Lambda}\mathbf{Z}_k$ is related to only one of the $\mathbf{x}_{v,k}$ \cite{Li2010,Li2012}. For the case $D=3$, we can write
\begin{align}\label{decoding1}
    \mathbf{Y}_k = \begin{bmatrix}
    \lambda_1\mathbf{g_1}^T\mathbf{x}_{1,k} & \lambda_1\mathbf{g_1}^T\mathbf{x}_{2,k} & \lambda_1\mathbf{g_1}^T\mathbf{x}_{3,k}\\
   g \lambda_2\mathbf{g_2}^T\mathbf{x}_{3,k} & \lambda_2\mathbf{g_2}^T\mathbf{x}_{1,k} & \lambda_2\mathbf{g_2}^T\mathbf{x}_{2,k} \\
    g \lambda_3\mathbf{g_3}^T\mathbf{x}_{3,k} & g\lambda_3\mathbf{g_3}^T\mathbf{x}_{1,k} & \lambda_3\mathbf{g_3}^T\mathbf{x}_{1,k}
\end{bmatrix} + \tilde{\mathbf{n}}_k.
\end{align}

The processed signal in (\ref{decoding1}) can be divided into $D$ parts. Then one can write
\begin{align}\label{rec_signal_tot}
    \mathbf{y}_{k,v} = \mathbf{\Omega}_v\mathbf{\Lambda}\mathbf{G}\mathbf{x}_{v,k}+\tilde{\mathbf{n}}_{k,v}
\end{align}
where $v=1,\dots,D$ and $\mathbf{\Omega}_v = \text{diag}(\omega_{v,1},\dots,\omega_{v,D})$. The elements of the matrix $\mathbf{\Omega}_v$ are defined as 
\begin{align}
\omega_{v,u} = \left\{
    \begin{array}{ll}
        1, \; \; \; 1 \leq u \leq D-v+1  \\
        g, \; \; \;  D-v+2 \leq u \leq  D.
    \end{array}
\right.
\end{align}

One can simplify (\ref{rec_signal_tot}) by using the QR decomposition of $\mathbf{\Lambda}\mathbf{G}=\mathbf{Q}\mathbf{R}$ as done in \cite{Li2012} to simplify the ML bit metrics defined in (\ref{ML_metric}). 
		\section{Simulation Results}
		In the simulations, the industry standard 64-state 1/2-rate (133,171) $d_{\text{free}}=10$ convolutional code is used. For BICMB, we separate the coded bits into different substreams of data and a random interleaver is used to interleave the bits in each substream. We assume that the number of RF chains in the receiver and transmitter are twice the number of data streams \cite{Sohrabi2016} (i.e., $N_t^{\text{RF}}=N_r^{\text{RF}}=2N_s$). Also, each scale fading coefficient $\beta_{ij}$ equals $\beta = -20$ dB for all simulations, except for Fig. \ref{diff_beta}. At RAUs in both transmitter and receiver, ULA array configuration with $d=0.5$ is considered. Information bits are mapped onto 16-QAM symbols in each subchannel.
		
		\begin{figure}[!t]
        	\centering
        	\includegraphics[width=.5\textwidth]{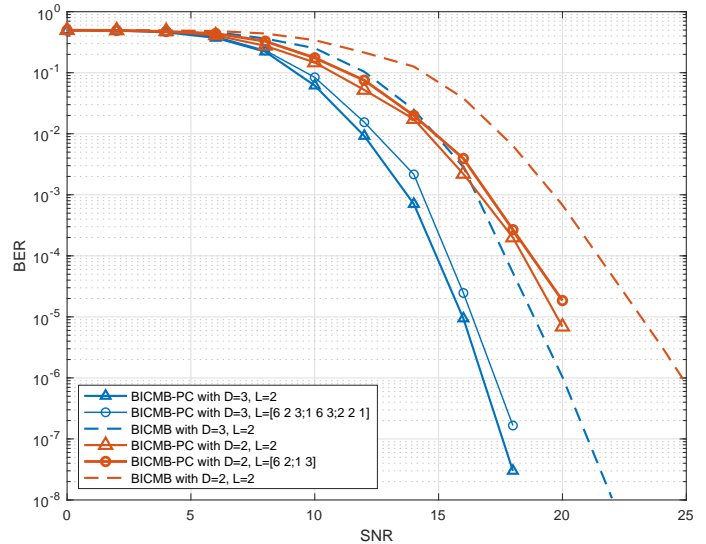}
	         \caption{BER with respect to SNR for setups. $N_t=128$ and $N_r=64$.}
	        \label{ber_pc}
        \end{figure}

		Fig. \ref{ber_pc} illustrates the results for BICMB perfect coding (BICMB-PC) for different values of $D$ and $L_{ij}$ in a mm-wave massive MIMO system. Furthermore, we can see the comparison of the BICMB-PC with the BICMB results  in \cite{Sedighi2019}. We define the number of propagation paths as $\mathbf{L}=[L_{11}\; L_{12};L_{21} \; L_{22}]$ and $\mathbf{L}=[L_{11}\; L_{12}\; L_{12}; L_{21} \; L_{22} \;L_{23};L_{31} \; L_{32} \;L_{33}]$. When 
		$\mathbf{L}=l$, all elements in $\mathbf{L}$ are constant and equal to $l$. It can be seen that the diversity gain remains the same for different values for the number of propagation paths, as long as (\ref{su_Gd}) returns the same value of $G_d$. For example for the red solid line curves with markers, since $\beta_{ij}=\beta$, the diversity gains are $2 \times 2 \times 2 = (2\times 2)^2 / (6^{-1}+2^{-1}+3^{-1}+1^{-1})$ as in (\ref{su_Gd}). Same can be applied for the blue solid line curve where $D=3$. Furthermore, it can be seen that the BICMB curves in \cite{Sedighi2019}, has the same slope in high SNR, i.e., same diversity gain as the BICMB-PC for different setups.
				\begin{figure}[!t]
        	\centering
        	\includegraphics[width=.5\textwidth]{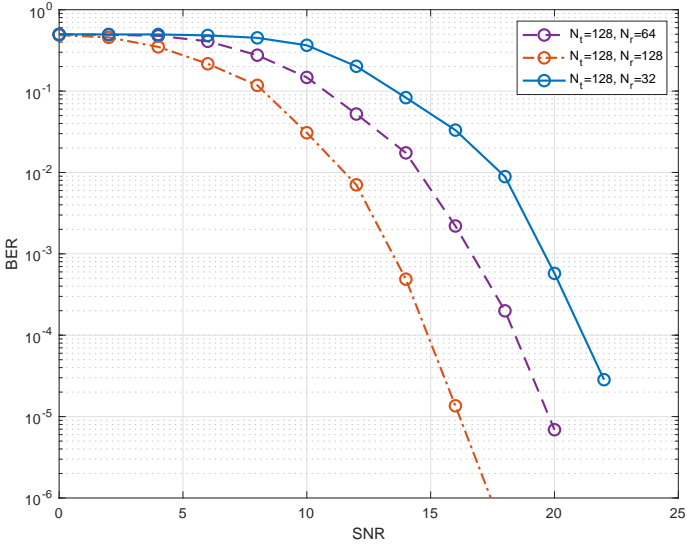}
        	\caption{BER with respect to SNR for different number of antennas at each RAU at the transmitter and receiver. $D=2$ and $L=2$.}
	        \label{dif_anten}
        \end{figure}

		It can be seen from Fig. \ref{dif_anten} that changing the number of antennas at the RAUs does not affect the diversity gain. This confirms (\ref{su_Gd}) where the diversity gain is independent of the number of antennas at each RAU. Furthermore, one can see that by doubling the number of resources here, i.e., the number of antennas at the transmitter or the receiver, the performance of the system gets better by a factor of 3 dB.

				\begin{figure}[!t]
        	\centering
        	\includegraphics[width=.5\textwidth]{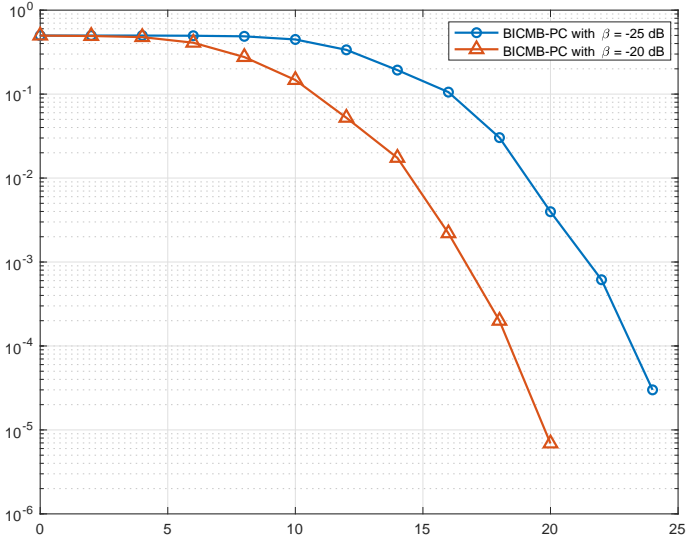}
        	\caption{BER with respect to SNR for different values of large scale fading coefficients, i.e., $\beta$. $D=2$ and $L=2$.}
	        \label{diff_beta}
        \end{figure}
        
        Fig. \ref{diff_beta} shows that by having different values for the large scale fading coefficient, i.e., $\beta$, the diversity does not change, as long as this coefficient remains the same for all subchannels.

		\section{Conclusion}
	In this work we showed that by utilizing both the BICMB and perfect coding in a mm-wave massive MIMO system with distributed antenna-subarray architecture, one can achieve both full diversity gain and full spatial multiplexing gain. This means, the diversity gain is independent from the number of transmitted data streams and can be increased by increasing the number of RAUs at both the transmitter and the receiver. We also show that the diversity gain is independent from the number of antennas at the RAUs in both the transmitter and the receiver. 
	\bibliographystyle{IEEEtran}
	\bibliography{perfect}
	\end{document}